\documentclass[11pt]{article}
\usepackage{amssymb,amsfonts,amsmath,amsthm}%
\usepackage{graphicx,color}%
\newcommand{\figdraft}{false}%
\newcommand{\figfile}[1]{#1}%
\theoremstyle{plain}%
\newtheorem{theorem}{Theorem}[]%
\newtheorem{corollary}[theorem]{Corollary}%
\newtheorem{lemma}[theorem]{Lemma}%
\newtheorem{assumption}[theorem]{Assumption}%
\theoremstyle{definition}
\newtheorem{remark}[theorem]{Remark}%
\newcommand{\dblidx}[2]{_{#1,\,#2}}

\newcommand{\dblprm}[2]{{#1,\,#2}}

\newcommand{\iu}{\mathtt{i}}
\newcommand{\mhexp}[1]{{{\mathtt{e}}^{#1}}}

\newcommand{\phase}{{\varphi}}

\newcommand{\Rset}{{\mathbb{R}}}
\newcommand{\Zset}{{\mathbb{Z}}}

\newcommand{\cointerval}[2]{[#1,\,#2)}%
\newcommand{\oointerval}[2]{(#1,\,#2)}%
\newcommand{\ccinterval}[2]{[#1,\,#2]}%

\newcommand{\Do}[1]{{o\at{#1}}}
\newcommand{\nDO}[1]{{O\nat{#1}}}
\newcommand{\nDo}[1]{{o\nat{#1}}}

\newcommand{\ini}{{\rm ini}}

\newcommand{\tdots}{{...}}%

\newlength{\mhpicDwidth}
\newlength{\mhpicDvsep}
\newlength{\mhpicDhsep}

\newlength{\mhpicPwidth}
\newlength{\mhpicPvsep}
\newlength{\mhpicPhsep}

\newlength{\mhpicWhsep}

\newcommand{\pair}[2]{{\left({#1},\,{#2}\right)}}
\newcommand{\skp}[2]{{\left\langle{#1},\,{#2}\right\rangle}}

\newcommand{\bpair}[2]{{\big({#1},\,{#2}\big)}}

\newcommand{\at}[1]{{\left({#1}\right)}}
\newcommand{\nat}[1]{(#1)}
\newcommand{\bat}[1]{{\big(#1\big)}}

\newcommand{\ul}[1]{\underline{#1}}

\newcommand{\T}{\textstyle}

\newcommand{\bigpar}{\par\quad\newline\noindent}

\newcommand{\norm}[1]{\|{#1}\|}
\newcommand{\abs}[1]{\left|{#1}\right|}
\newcommand{\babs}[1]{\big|{#1}\big|}

\newcommand{\dint}[1]{\,\mathrm{d}#1}

\newcommand{\al}{{\alpha}}

\newcommand{\eps}{{\varepsilon}}

\newcommand{\ka}{{\kappa}}
\newcommand{\la}{{\lambda}}
\newcommand{\si}{{\sigma}}

\newcommand{\calC}{\mathcal{C}}

\newcommand{\calE}{\mathcal{E}}
\newcommand{\calF}{\mathcal{F}}

\newcommand{\calH}{\mathcal{H}}

\newcommand{\calL}{\mathcal{L}}

\newcommand{\calN}{\mathcal{N}}

\newcommand{\calP}{\mathcal{P}}

\newcommand{\calS}{\mathcal{S}}

\newcommand{\calW}{\mathcal{W}}

\usepackage[%
a4paper,%
nohead,%
twoside,
ignorefoot,
scale=0.8,
bindingoffset=.4cm
]{geometry}%

\usepackage{float}%
\begin{document}%
%
%
\title{Homoclinic standing waves in focusing DNLS equations%
\\\emph{Variational approach via constrained optimization}%
}%
\date{\today}%
\author{%
Michael Herrmann%
\thanks{
    Department of Mathematics, Saarland University,
    66123 Saarbr\"ucken, Germany \newline EMail:
    {\tt{michael.herrmann@math.uni-sb.de}}
}%
}%
\maketitle
%
%
%
\begin{abstract}%
We study focusing discrete nonlinear Schr\"{o}dinger equations and present a new variational
existence proof for homoclinic standing waves (bright solitons). Our approach relies on the
constrained maximization of an energy functional and provides the existence of two
one-parameter families of waves with unimodal and even profile function for a wide class of
nonlinearities. Finally, we illustrate our results by numerical simulations.
\end{abstract}%
%
%
\quad\newline\noindent%
\begin{minipage}[t]{0.15\textwidth}%
Keywords: %
\end{minipage}%
\begin{minipage}[t]{0.8\textwidth}%
\emph{discrete nonlinear Schr\"{o}dinger equation (DNLS)}, \\%
\emph{homoclinic standing waves}, %
\emph{solitary waves}, %
\emph{bright solitons}, %
\emph{breathers}, \\%
\emph{nonlinear lattice waves, constrained maximization} %
\end{minipage}%
\medskip
\newline\noindent
\begin{minipage}[t]{0.15\textwidth}%
MSC (2000): %
\end{minipage}%
\begin{minipage}[t]{0.8\textwidth}%
37K60, 
47J30, 
78A40 
\end{minipage}%
%
%
%
%
%

%
%

\section{Introduction}
%
The discrete nonlinear Schr\"{o}dinger equation (DNLS) is one of the most fundamental lattice
equations and plays a prominent role in the theories of optical waveguides, photorefractive
crystals, and Bose-Einstein condensates. For an overview on applications and results we refer
the reader to \cite{KRB01,EJ03,Kev:DNLS:Ch1,Por09} and the references therein. In this paper
we aim in contributing to the mathematical theory of DNLS by giving a new variational
existence proof for homoclinic standing waves.
\bigpar
In one space dimension the homogeneous DNLS is given by
\begin{align}%
\label{Eqn:DNLS0}
\iu\dot{A}_j+\al\at{A_{j+1}+A_{j-1}-2A_j}+\beta{A_j}+\Psi^\prime\nat{\abs{A}^2_j}A_j=0,
\end{align}
where $j$ is the discrete space variable, $t$ is the time, and $A_j=A_j\at{t}$ is the
complex-valued amplitude or the value of the wave function. The potential function $\Psi$ is
often assumed to be monomial with $\Psi^\prime\nat{x}=x$ and $\Psi^\prime\nat{x}=x^2$ for
the cubic and quintic DNLS, respectively.
\par
For our purposes it is more convenient to simplify the linear terms by means of the gauge
invariance of DNLS, that is the symmetry under $u_j\rightsquigarrow\mhexp{\iu\phase}u_j$ with
$\phase\in\Rset$. More precisely, with $A_j\rightsquigarrow\mhexp{\iu\at{\beta-2\alpha}t}A_j$
we readily verify that \eqref{Eqn:DNLS0} is equivalent to

\begin{align}
\label{Eqn:DNLS1}
\iu\dot{A}_j+\al\at{A_{j+1}+A_{j-1}}+\Psi^\prime\nat{\abs{A}^2_j}A_j=0.
\end{align}
It is well established that the dynamical properties of \eqref{Eqn:DNLS1} strongly depend on
the sign and strength of the coupling parameter $\alpha$. For convex $\Psi$, one usually
refers to $\alpha>0$ and $\alpha<0$ as the \emph{focusing} and \emph{defocusing} case,
respectively. In the anti-continuum limit $\alpha\to0$ the DNLS becomes a system of uncoupled
oscillators, whereas in the continuum limit $\alpha\to\pm\infty$ the DNLS can be viewed as a
finite difference approximation of the nonlinear Schr\"{o}dinger PDE (via the scaling
$j\rightsquigarrow\sqrt{\abs{1/\alpha}}j$).
\bigpar
During the last decades a great deal of attention has been paid to coherent structures such
as travelling waves and standing waves. Standing waves can be regarded as relative equilibria
which stem from the gauge invariance. They are special solutions to \eqref{Eqn:DNLS1} with
$A_j\at{t}=\mhexp{\iu\si{t}}u_j\at{t}$, where the profile $u=\at{u_j}_j$ is assumed to take
values in $\Rset$ and satisfies
\begin{align}%
\label{Eqn:StandingWaveEquation}
\si{u}_j=\al\at{u_{j+1}+u_{j-1}}+\Psi^\prime\at{u_j^2}u_j.
\end{align}
Standing waves come in different types: Periodic waves (or \emph{wave trains}) satisfy
$u_j=u_{j+N}$ for some periodicity length $N<\infty$. Homoclinic waves (\emph{solitons}, \emph{solitary waves}) are localized via
$\lim_{j\to\pm\infty}u_j=0$, and are hence also \emph{breathers}.
Finally, heteroclinic waves (\emph{fronts}, \emph{kinks}) connect different asymptotic states.
\par

The existence of standing wave solutions to \eqref{Eqn:DNLS1} has been investigated by
several authors using rather different methods. Sometimes it is possible to find exact
solutions, see \cite{ELS85,KRSS05}, but in general one needs more sophisticated and robust
arguments. A perturbative approach to the existence problem was developed by MacKay and Aubry
\cite{MA94,Aub97} and relies on continuation method. The main idea is to start with a given
solution in the anti-continuum limit $\alpha=0$ and to show that there is a corresponding solution for small
$\alpha$. Continuation arguments have been proven powerful for both analytical
considerations and numerical simulations and seem to be the preferred method in the physics
community.
\par
The main limitation of any continuation methods, however, is the need of an anchor solution
around which the equation is expanded. As a consequence there is a growing interest in
alternative existence proofs for standing waves. Example are dynamical systems approaches \cite{QX07,PR05},
or variational methods that employ critical point techniques (linking theorems, Nehari
manifolds) to establish the existence of waves with prescribed frequency $\si$, see 
\cite{PZ01,PR08,ZP09,ZL09} and 
\cite{Pan06,Pan07,SZ10} for similar results in DNLS with periodic coefficients.
%
\subsection{Variational setting}
%
In this paper we rely on a variational setting, which does not prescribe the frequency
but the power of a standing wave. More precisely, we obtain homoclinic standing waves as solutions to a 
constrained optimization problem, in which the frequency $\si$ is the Lagrangian multiplier. 
A similar idea was used by Weinstein \cite{Wei99} for DNLS with power nonlinearities, 
but we allow for a wider class of nonlinear potentials
$\Psi$. Moreover, our approach provides more information about the shape
of standing waves as it guarantees the existence of waves with unimodal and even profile
$u$. We also emphasize that, contrary to variational methods with prescribed $\si$, our existence proof gives rise to
an effective approximation scheme for standing waves. Finally, the restriction to
the one-dimensional case is not essential but was made for the sake of simplicity.
\par
In oder to sketch the main idea of our method we introduce an energy functional $\calP$ and
the power functional $\calN$ by
\begin{align}
\label{Eqn:Intro.Energies}
\calP\at{u}=\sum\limits_j\Psi\at{u_j^2}+\alpha\sum\limits_j\at{u_{j+1}+u_{j-1}}u_j
,\qquad\calN\at{u}=\sum\limits_j{u_j^2}.
\end{align}
Both $\calN$ and $\calP$ are related to conserved quantities for the Hamiltonian system
\eqref{Eqn:DNLS0}. In fact, $\calN$ is linked to the gauge invariance by Noether's Theorem, and
rearranging the quadratic terms we find
$\calP\at{u}=2\alpha\calN\at{u}-\calH\at{u}$, where
\begin{align*}
\calH\at{A}=\sum\limits_j{\alpha\babs{A_{j+1}-A_j}^2}-\Psi\nat{\abs{A_j}^2}
\end{align*}
is the Hamiltonian corresponding to \eqref{Eqn:DNLS0}. We readily verify that the standing
wave equation \eqref{Eqn:StandingWaveEquation} is equivalent to
\begin{align}
\label{Eqn:Intro.EULA}
\si\partial\calN\at{u}=\partial\calP\at{u}
\end{align}
where $\partial$ denotes the variational derivative with respect to $u$. The key observation
is that \eqref{Eqn:Intro.EULA} can be considered as the Euler--Lagrange equation of  the
optimization problem
\begin{align}
\label{Eqn:Intro.OP}
\text{maximize $\calP\at{u}$ under the constraint $\calN\at{u}=\varrho$},
\end{align}
where $\si$ plays the role of an Lagrangian multiplier. Notice that \eqref{Eqn:Intro.OP} is equivalent to 
minimizing the energy $\calH$ subject to prescribed power $\calN$, which is a well established idea in the theory of standing waves, see \cite{Wei99} for DNLS and \cite{Pav09,Stu09} for dispersive Hamiltonian PDEs.
\par
In this paper we refine the optimization problem \eqref{Eqn:Intro.OP} by considering only
those profiles $u$ that are non-negative, unimodal, and even. Specifically, we solve the optimization problem
\begin{align}
\label{Eqn:Intro.OPRef}
\text{maximize $\calP\at{u}$ under the constraints $\calN\at{u}=\varrho$ and $u\in\calC$},
\end{align}
where the convex cone $\calC$ consists of all profiles $u$ that satisfy
$u_{-j}=u_{j}$ and $u_{j}\geq{u_{j+1}}\geq0$ for all $j\geq0$.
Of course, we then have to show that
each solution to the so restricted optimization problem satisfies the standing wave equation
\eqref{Eqn:Intro.Energies} without further multipliers.
\par
It is known that there exist two possible choices for the index $j$. In the \emph{on-site}
(or \emph{site-centered}) setting we suppose $j\in{\Zset}$, whereas in the \emph{inter-site}
(or \emph{bond-centered}) setting we choose $j\in\Zset+\tfrac{1}{2}$. Both settings are
equivalent on the level of \eqref{Eqn:DNLS1}, and even on the level of \eqref{Eqn:Intro.OP}
with $u\in\ell^2$, but lead to different results when studying waves with even and unimodal
profile $u\in\calC$. In fact, on-site waves attain their maximum in an odd number of points centered
around $j=0$ (generically only in $j=0$), whereas the maximum of inter-site waves is realized
in an even number of points (generically only in $j=-\tfrac{1}{2}$ and $j=\tfrac{1}{2}$).
%
\subsection{Sketch of the proof and main result}
%
%
Due to the lack of strong compactness, is not trivial to show that $\calP$
attains its maximum on the set of interest. A standard strategy would be to
employ Lion's \emph{concentration compactness principle} \cite{Lio84}, see also \cite{Wei99,Pav09}, but
 we argue differently: At first we
consider the analogue to \eqref{Eqn:Intro.OPRef} in the space of periodic profiles with $u_j=u_{j+N}$.
The existence of a maximiser is then granted and 
the invariance properties of the reversed gradient flow for $\calP$ ensure that the maximizer
solves \eqref{Eqn:Intro.EULA} with some multiplier $\si$. In particular, there are no
multipliers due to the unimodality constraint, and so we obtain periodic waves with unimodal and even profile.
In the second step we then establish the existence of homoclinic waves by
passing to the limit $N\to\infty$. To this end we exploit a \emph{strict maximum condition},
which replaces the concentration compactness principle and guarantees that the periodic waves are uniformly localized. We mention that approximation 
by periodic waves is also used in \cite{Pan06,Pan07}, but in a variational setting that prescribes $\si$ and uses critical techniques to construct saddle points of an action integral.
\bigpar
Our main result on standing waves for DNLS can be summarized as follows.
\begin{theorem}
\label{Intro:Result}%
 Suppose $\alpha>0$ and that $\Psi$ satisfies the super-linear growth and
regularity conditions formulated in Assumption \ref{Intro:Ass}. Then, in both the on-site
and the inter-site setting there exists a one-parameter family of homoclinic standing waves
$\pair{u}{\si}$ that is parametrized by
$\varrho=\norm{u}^2\geq{\varrho_\ast}\at\alpha\geq0$. For each wave, the frequency satisfies
$\si>2\alpha$ and the profile $u$ is non-negative, even, unimodal, and exponentially decaying.
\end{theorem}
We proceed with some remarks concerning the assumptions and assertions of Theorem
\ref{Intro:Result}.
\begin{enumerate}
\item
The class of admissible $\Psi$ includes all convex functions with $\Psi^\prime\at{0}=0$. In particular, it contains power nonlinearities
\begin{align*}
\Psi\at{x}=\frac{1}{\eta}x^{1+\eta},\quad \Psi^\prime\at{x}=x^\eta,\quad \eta>0,
\end{align*}
but also potentials with saturable derivative such as
\begin{align*}
\Psi\at{x}=x-\log\at{1+x}
,\quad
\Psi^\prime\at{x}=\frac{x}{1+x}.
\end{align*}
Notice that the dynamical systems approach by Qin and Xioa \cite{QX07} also allows for arbitrary strictly convex $\Psi$ and guarantees the existence of a homoclinic wave for all $\si\in\oointerval{2\alpha}{2\alpha+h_\infty}$ with $h_\infty=\lim_{x\to\infty}\Psi^\prime\at{x}$.
\item Theorem \ref{Intro:Result} also covers some non-convex potentials as for
instance
\begin{align*}
\Psi\at{x}=\frac{x^3}{{1+x^2}}
,\quad
\Psi^\prime\at{x}=\frac{x^2\at{3 + x^2}}{\at{1 + x^2}^2}.
\end{align*}
\item
Our results derived below imply $\varrho_\ast\at{\alpha}\to0$ as $\alpha\to0$.
Therefore, instead of fixing $\alpha$ and choosing $\varrho$ sufficiently large we can
alternatively fix $\varrho$ and choose $\alpha$ sufficiently small. 
We also derive explicit criteria for $\Psi$, which
guarantee that $\varrho_\ast\at{\alpha}=0$ and reflect
well known results about the \emph{excitation threshold} for breathers in 1D, see
\cite{FKM97,Wei99,DZC08}.
\item
Theorem~\ref{Intro:Result} provides the existence of \emph{bright solitons} for
$\alpha>0$. Via the \emph{staggering transformation} $u_j\rightsquigarrow\at{-1}^j{u_j}$
it also implies the existence of standing waves for the defocusing case $\alpha<0$.
These resulting waves have an alternating phase structure and are called \emph{bright gap
solitons}, see \cite{Kev:DNLS:Ch5}. However, the most fundamental standing waves for
$\alpha<0$ are \emph{dark solitons} which correspond to heteroclinic solutions to
\eqref{Eqn:StandingWaveEquation} and require a different variational setting. This is
discussed in \cite{Her10PreD}.
\end{enumerate}
We do not claim that the constrained maximization of $\calP$ provides all standing wave
solutions to \eqref{Eqn:DNLS1}. It fact, by continuation from the anti-continuum limit it
is known that there exist infinitely many \emph{multipulse solitons}, see
\cite{Kev:DNLS:Ch2}. However, the careful spectral analysis by Pelinovsky et al. \cite{PKF05} indicates that most of these multipulse solitons are unstable. This is true even for unimodal inter-site waves, but at least unimodal on-site waves can be expected
to be stable. We also emphasize that maximisers of the optimization problem
\eqref{Eqn:Intro.OP} can be shown to be orbitally stable \cite{Wei86,Wei99}, for instance using the Grillakis--Shatah--Strauss theory, see \cite{Pav09} and the references therein. Due to the additional shape constraint $u\in\calC$, however, this stability
result does not apply directly to the waves provided by Theorem~\ref{Intro:Result} because it is not clear that the optimization problems \eqref{Eqn:Intro.OP} and \eqref{Eqn:Intro.OPRef} have the same solution. We conjecture that this is indeed the case but a rigorous proof is still missing.
\bigpar
This paper is organized as follows. In \S\ref{sec:Prelim} we formulate our assumptions on
$\Psi$ and introduce some notations. We also introduce the strict maximum condition and
investigate the maximum of $\calP$ on bounded subsets of $\ell^2$. In \S\ref{sec:Waves} we
then proof the existence of periodic waves and pass to the limit $N\to\infty$. Finally, we
present some numerical simulations in \S\ref{sec:ApproxWaves}.
%
%
%
\section{Setting and properties of the energy landscape}\label{sec:Prelim}
%
%
In this paper we always assume $\al>0$ and rely on the following standing assumptions on
$\Psi$.
\begin{assumption}
\label{Intro:Ass}
The potential $\Psi$ is continuously differentiable on $\cointerval{0}{\infty}$ and has the
following properties.
\begin{enumerate}
\item
$\Psi$ is normalized by $\Psi^\prime\at{0}=\Psi\at{0}=0$.
\item
$\Psi$ is grows super-linearly, that means $x\Psi^\prime\at{x}\geq\Psi\at{x}\geq0$ for all
$x\geq0$.
\item
$\Psi$ is non-degenerate, i.e., $\Psi\at{x}>0$ for $x>0$.
\end{enumerate}
\end{assumption}
\begin{remark}
\quad
\begin{enumerate}
\item
Assumption \ref{Intro:Ass} implies
\begin{align}
\label{Eqn:PotProps1}
0\leq{x^{-1}}\Psi\at{x}\leq\Psi^\prime\at{x}\xrightarrow{x\to0}0,\quad
\Psi\at{\la{x}}\geq\la\Psi\at{x}\;\;\;\forall\;x\geq0,\;\la\geq1,\quad
\Psi\at{x}\xrightarrow{x\to\infty}\infty.
\end{align}
\item
Each convex and normalized potential 
$\Psi$ grows super-linearly since the function $\theta\at{x}=x\Psi^\prime\at{x}-\Psi\at{x}$ is
non-negative due to $\theta\at{0}=0$ and $\theta^\prime\at{x}=x\Psi^{\prime\prime}\at{x}\geq0$.
\item Assumption \ref{Intro:Ass} is invariant under scalings
\begin{align*}
\Psi\rightsquigarrow\widehat{\Psi}, \qquad \widehat{\Psi}\at{x}:={a}\Psi\at{bx^{1+c}},
\end{align*}
where $a,b,c>0$ are given constants. This means $\Psi$ satisfies Assumption \ref{Intro:Ass} if and only if
$\widehat{\Psi}$ does so.
\end{enumerate}
\end{remark}
%
\subsection{Notations}
%
%
In what follows we consider real-valued sequences $u=\at{u_j}_{j\in{J}}$, called
\emph{profiles}, where the index set is given by $J=\Zset$ and $J=\Zset+\tfrac{1}{2}$ in the
on-site and inter-site setting, respectively. Under the periodicity condition $u_j=u_{j+N}$
for all $j$ and some $N$, each profile $u$ is uniquely determined by its values on a
periodicity cell $Z_N$. We choose
\begin{align*}
Z_{2M}=\{-M+1,\,\tdots,\,-1,\,0,\,1,\,\tdots,\,M\}
,\quad%
Z_{2M+1}=\{-M,\,\tdots,\,-1,\,0,\,1,\,\tdots,\,M\},
\end{align*}
in the on-site setting, and
\begin{align*}
Z_{2M}
=%
\{-M+\tfrac{1}{2},\,\tdots,\,-\tfrac{1}{2},\,\tfrac{1}{2},\,\tdots,\,M-\tfrac{1}{2}\}
,\quad%
Z_{2M+1}
=%
\{-M+\tfrac{1}{2},\,\tdots,\,-\tfrac{1}{2},\,\tfrac{1}{2},\,\tdots,\,M+\tfrac{1}{2}\}
\end{align*}
in the inter-site setting. The symmetrized periodicity cell is abbreviated by
$\tilde{Z}_N={Z_N}\cap\at{-{Z_N}}$ and we readily verify, see Figure \ref{Fig:WaveTypes}, that
\begin{align}
\label{Eqn:Misc1}
\at{N-1}/2\leq\max{Z_N}\leq{N/2},\qquad\#\big\{i\in{Z_N}\;:\;\abs{i}
\leq%
\abs{j}\,\big\}
=%
2\abs{j}+1
\quad\forall\;{j\in{\tilde{Z}_N}}.
\end{align}
\begin{figure}[ht!]
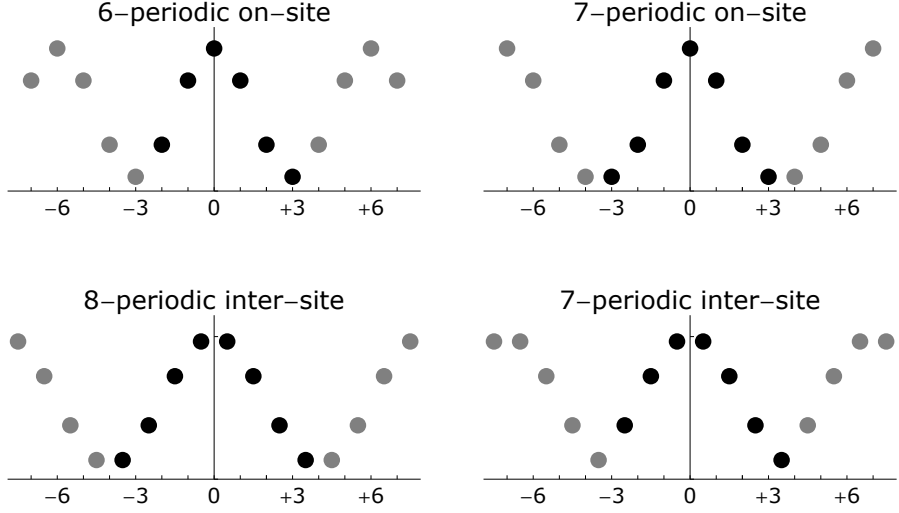
%
\centering{%
\includegraphics[width=0.33\textwidth, draft=\figdraft]%
{\figfile{wt_on_even}}%
\hspace{0.05\textwidth}%
\includegraphics[width=0.33\textwidth, draft=\figdraft]%
{\figfile{wt_on_odd}}%
\vspace{0.05\textwidth}\\%
\includegraphics[width=0.33\textwidth, draft=\figdraft]%
{\figfile{wt_inter_even}}%
\hspace{0.05\textwidth}%
\includegraphics[width=0.33\textwidth, draft=\figdraft]%
{\figfile{wt_inter_odd}}%
\caption{%
Examples for periodic profiles that are non-negative, even, and unimodal.
The values in the periodicity cell are plotted in Black.
}%
\label{Fig:WaveTypes}%
}%
\end{figure}%
For our purposes it is convenient to identify the non-periodic case with $N=\infty$, so we
write $Z_\infty=\tilde{Z}_\infty=J$. In view of \eqref{Eqn:Intro.Energies} and
\eqref{Eqn:Intro.EULA} we introduce the discrete function spaces
\begin{align*}%
\ell^2_N&=\Big\{u=\at{u_j}_{j\in{J}}\;:\;u_{j}=u_{j+N}
\;\text{for all $j\in{J}$},\;\sum_{j\in{Z_N}}u_j^2<\infty\Big\}
\cong\Rset^{N},
\\%
\ell^2_\infty&=
\Big\{u=\at{u_j}_{j\in{J}}\;:\;\sum_{j\in{Z_\infty}}u_j^2<\infty\Big\}
\cong\ell^2\at{\Zset},
\end{align*}
for $N<\infty$ and $N=\infty$, respectively. For both finite and infinite $N$ we denote by
\begin{align*}%
\norm{u}^2=\sum_{j\in{Z_N}}u_j^2
,\qquad
\skp{u}{v}=\sum_{j\in{Z_N}}u_jv_j,
\end{align*}
the norm and scalar product in $\ell^2_N$, so the sphere of radius $\sqrt{\varrho}$ reads
\begin{align*}
\calS_\dblprm{N}{\varrho}=\{u\in\ell^2_N\;:\;\norm{u}^2=\varrho\}.
\end{align*}
We also define the energy functional $\calP_N$ by
\begin{align*}
\calP_N\at{u}=\alpha\calL_N\at{u}+\calW_N\at{u}
,\qquad
\calL_N\at{u}=\sum_{j\in{Z_N}}{u_j}\at{u_{j+1}+u_{j-1}},\qquad
\calW_N\at{u}=\sum_{j\in{Z_N}}\Psi\nat{u_j^2},
\end{align*}
and denote by
\begin{align*}
\calC_N=
\left\{u\;:\;u_j=u_{-j}\geq0\;\;\;\forall\;\;j\in{Z_N}\right\}
\,\cap\,%
\left\{u\;:\;u_{j-1}\geq{u_{j}}\;\;\;\forall\;\;1\leq{}j\in{Z_N}\right\}
\end{align*}
the set of all profiles $u$ that are non-negative, even and unimodal on the periodicity cell
$Z_N$.
\par
From these definitions we readily draw the following conclusions.
\begin{lemma}
\label{Lem:BasicProps}
For both finite and infinite $N$ the following assertions are satisfied.
\begin{enumerate}
\item
$\calP_N$ is G\^{a}teaux-differentiable on $\ell^2_N$ with derivative
$\partial\calP\at{u}_j=2\alpha\at{u_{j+1}+u_{j-1}}+2\Psi^\prime\nat{u_j^2}u_j$.
\item
$\calC_N$ is a positive convex cone and closed under weak
convergence in $\ell^2_N$.
\item
By \eqref{Eqn:Misc1} we have
\begin{math}
u_j\leq\norm{u}/\sqrt{2\abs{j}+1} %
\end{math}
for all $u\in\calC_N$ and $j\in\tilde{Z}_N$.
\end{enumerate}
\end{lemma}
The core of our variational existence proof for standing waves is the optimization problem
\begin{align}
\label{Intro:OptProb}
\text{maximize $\calP_N$ on the set $\calC_N\cap\calS_\dblprm{N}{\varrho}$}.
\end{align}
For finite $N$, the existence of maximizers follows from simple compactness arguments and the
invariance properties of the gradient flow of $\calP_N$ imply that each maximizer is in fact a periodic standing wave. For infinite $N$, however, we lack compactness and construct homoclinic
maximizers of $\calP_\infty$ as limits of periodic maximizers of $\calP_N$.
%
%
%
\subsection{Energy maxima on bounded sets}
%
In this section we investigate the function 
\begin{align}
\notag
T_N\pair{\alpha}{\varrho}=\frac{1}{\alpha\varrho}
\sup\left\{\calP_N\at{u}\;:\:u\in\calC_N\cap\calS_\dblprm{N}{\varrho}\right\},
\end{align}
and show that the \emph{strict maximum condition}
\begin{align}
\label{Cond:TSC}
T_\infty\pair{\alpha}{\varrho}>2
\end{align}
holds if one of the following conditions is satisfied:
\begin{enumerate}
\item[$\at{A1}$]
$\alpha$ is sufficiently small,
\item[$\at{A2}$]
$\varrho$ is sufficiently large,
\item[$\at{A3}$]
$\Psi\at{x}\sim{cx}^{1+\eta}$ for all $0\leq{x}\ll1$ and $0<\eta<2$ and some $c>0$.
\end{enumerate}
The strict maximum condition  \eqref{Cond:TSC} appears naturally in our existence proof for homoclinic waves,
see \S\ref{sec:HomoclinicWaves}, and guarantees that the influence of the nonlinearity $\Psi$
is strong enough in the limit $N\to\infty$. Specifically, if \eqref{Cond:TSC} is satisfied,
then the optimization problem \eqref{Intro:OptProb} has a solution for $N=\infty$, and each minimizer is a homoclinic wave with 
unimodal profile. Moreover, if \eqref{Cond:TSC} is violated, then we have
$T_\infty\pair{\alpha}{\varrho}=2=\sup_{u\in\calS_\dblprm{\infty}{1}}\calL_\infty\at{u}$, and this implies
that \eqref{Intro:OptProb} has no solution for $N=\infty$.
\par

We also emphasize that \eqref{Cond:TSC} is equivalent to
\begin{align*}
\inf\Big\{\calH\at{u}\:;\:u\in\ell^2\at{\Zset},\;\calN\at{u}=\varrho\Big\}<0,
\end{align*}
which is precisely the condition used in \cite{Wei99} to prove the existence of homoclinic standing waves for DNLS with power nonlinearity $\Psi\at{x}=cx^{1+\eta}$. Moreover, condition $\at{A3}$ implies that 
there is no excitation threshold for power nonlinearities with $0<\eta<2$,
which means there exist waves with arbitrary small $\varrho$. This  result is again in line with the findings from \cite{Wei99}.
\bigpar
We now summarize some elementary properties of the function $T_N$.
\begin{lemma}
\label{Lem:TProps1}%
$T_N$ is decreasing in $\alpha$ and increasing in $\varrho$ for both finite and infinite $N$.
\end{lemma}
\begin{proof}

The monotonicity with respect to $\alpha$ is obvious. Towards the monotonicity in $\varrho$
let $u\in\ell^2_N$ and $\la\geq1$ be fixed. We have $\calL_N\at{\la{u}}=\la^2\calL_N\at{u}$,
and
\begin{align*}
\tfrac{\dint}{\dint\la}\calW_N\at{\la{u}}=
\sum\limits_{j\in{Z}_N}2\Psi^\prime\at{\la^2u_j^2}\la{u_j^2}\geq
\sum\limits_{j\in{Z}_N}\frac{2}{\la}\Psi\at{\la^2u_j^2}=
\frac{2}{\la}\calW_N\at{\la{u}}
\end{align*}
implies $\calW_N\at{\la{u}}\geq\la^2\calW_N\at{u}$. Consequently, for $\norm{u}\neq0$ and
$\la>0$ we find
\begin{align*}
T_N\pair{\alpha}{\la^2\norm{u}^2}
\geq%
\frac{1}{\la^2\norm{u}^2}\at{\calL_N\at{\la{u}}+
\tfrac{1}{\alpha}\calW_N\at{\la{u}}}
\geq%
\frac{1}{\norm{u}^2}\at{\calL_N\at{u}+
\tfrac{1}{\alpha}\calW_N\at{u}},
\end{align*}
and taking the supremum over $u$ gives
$T_N\pair{\alpha}{\la^2\norm{u}^2}\geq{}T_N\pair{\alpha}{\norm{u}^2}$.
\end{proof}
\begin{remark}
For power nonlinearities $\Psi\at{x}=cx^{1+\eta}$ with $c,\,\eta>0$ we have
$T_N\pair{\alpha}{\la^2\varrho}=\T_N\pair{\la^{-2\eta}\al}{\varrho}$ and hence
\begin{align*}
T_N\pair{\alpha}{\varrho}=T_N\bpair{\al\varrho^{-\eta}}{1}=
T_N\bpair{1}{\alpha^{-1/\eta}\varrho}
\end{align*}
for all
$\al,\,\varrho,\,\la>0$ and both $N<\infty$ and $N=\infty$.
\end{remark}
For finite $N$, the constant profile with $u_{j}=\sqrt{\varrho/N}$ for all $j$ belongs to
$\calC_N\cap\calS_\dblprm{N}{\varrho}$, and using \eqref{Eqn:PotProps1} we readily
verify that
\begin{align}
\label{Eqn:TProps4}%
\calP_N\bat{\sqrt{\varrho/N}}
=%
2\alpha\varrho+\varrho\frac{N}{\varrho}\Psi\at{\varrho/N}
\xrightarrow{N\to\infty}%
2\alpha\varrho.
\end{align}
In particular, due to $\Psi\at{\varrho/N}>0$ we have
\begin{align}
\label{Eqn:TProps3}%
T_N\pair{\alpha}{\varrho}\geq\calP_N\bat{\sqrt{\varrho/N}}>2
\end{align}
for all $N<\infty$, $\al>0$, $\varrho>0$. The case $N=\infty$ is a bit more delicate.
\begin{lemma}
\label{Lem:TProps2a}%
We have $T_\infty\pair{\alpha}{\varrho}\geq2$ with strict inequality provided that $\at{A1}$
or $\at{A2}$ is satisfied.
\end{lemma}
\begin{proof}
We start with the proof in the on-site setting and define
$u_m\in\calC_\infty\cap\calS_\dblprm{\infty}{\varrho}$ by
\begin{align*}
u\dblidx{m}{j}
=%
\frac{\sqrt{\varrho}}{\sqrt{2m+1}}
\begin{cases}
1&\text{for $\abs{j}\leq{m}$},\\
0&\text{for $\abs{j}>{m}$}.
\end{cases}
\end{align*}
A direct calculation shows
\begin{align*}
\calL_\infty\at{u_m}=\frac{2m}{2m+1}\varrho
,\quad%
\calP_\infty\at{u_m}=\at{2m+1}\Psi\at{\frac{\varrho}{2m+1}},
\end{align*}
and hence
\begin{align}
\label{Lem:TProps2.Eqn1}
\begin{split}
T_\infty\pair{\alpha}{\varrho}
&\geq%
\frac{\al\calL_\infty\at{u_m}+\calP_\infty\at{u_m}}{\alpha\varrho}
\\&=%
2-\frac{1}{\varrho}\frac{\varrho}{2m+1}+
\frac{1}{\alpha}\frac{2m+1}{\varrho}\Psi\at{\frac{\varrho}{2m+1}}
\geq2-\frac{1}{\varrho}\frac{\varrho}{2m+1}.
\end{split}
\end{align}
This gives $T_\infty\pair{\alpha}{\varrho}\geq2$ by passing to the limit $m\to\infty$ and
using \eqref{Eqn:PotProps1}.
\par
We have now two possibilities to show $T_\infty\pair{\alpha}{\varrho}>2$. First, for
given $\varrho$ and $m$ we can choose $\alpha$ small. Second, for fixed $\alpha$ 
and any $\varrho>1$ we choose $m=m\at{\varrho}$ such that $2m-1<\varrho<2m+1$, so
\eqref{Lem:TProps2.Eqn1} gives
\begin{align*}
T_\infty\pair{\alpha}{\varrho}\geq2-\frac{1}{\varrho}+\frac{c}{\al}
\end{align*}
for some $c>0$, and the claim follows with $\varrho\to\infty$. 
\par
Finally, in the inter-site setting we define $u_m$ by
\begin{align*}
v\dblidx{m}{j}
=%
\frac{\sqrt{\varrho}}{\sqrt{2m}}
\begin{cases}
1&\text{for $\abs{j}\leq{m}-\tfrac{1}{2}$},\\
0&\text{for $\abs{j}>{m}-\tfrac{1}{2}$},\\
\end{cases}
\end{align*}
and argue analogously.
\end{proof}
We now show that the strict maximum conditions is always satisfied if $\Psi$ is sufficiently nonlinear for $x\approx{0}$.
The key idea in our proof is also used in \cite{Wei99} to 
disprove the existence of an excitation thresholds for power nonlinearities with $0<\eta<2$.
\begin{lemma}
\label{Lem:TProps2b}%
$\at{A3}$ implies $T_\infty\pair{\alpha}{\varrho}>2$.
\end{lemma}
\begin{proof}
We present the arguments for the on-site setting; the proof for 
the inter-site setting is similar. For fixed $\alpha>0$, $\varrho>0$ but arbitrary $0<\zeta\ll1$ we consider $u_\zeta\in\calC_\infty\cap\calS_\dblprm{\infty}{\varrho}$ defined by
\begin{align*}
u_{\zeta,\,j}=\sqrt{\varrho\tanh{\zeta}}\exp\at{-\zeta\abs{j}},
\end{align*}
which has small amplitudes due to $\tanh{\zeta}=\zeta+\nDO{\zeta^3}$.
By direct computations we find
\begin{align*}
\calL_\infty\at{u_\zeta}=\frac{2\varrho}{\cosh{\zeta}}=\varrho\at{2-\zeta^2+\nDO{\zeta^4}},
\end{align*}
as well as
\begin{align*}
\calW_\infty\at{u_\zeta}
&=%
\bat{1+\Do{1}}\sum_{j\in{J}_\infty}c{u_j}^{2+2\eta}=
\bat{1+\Do{1}}\,c\,\at{\varrho\tanh\zeta}^{1+\eta}\frac{1+\exp\bat{-2\zeta\at{1+\eta}}}{1-\exp\bat{-2\zeta\at{1+\eta}}}
\\&=\bat{1+\Do{1}}\frac{c}{1+\eta}\varrho^{1+\eta}\zeta^{\eta},%
\end{align*}
where $\Do{1}$ means arbitrary small for small $\zeta$. We now conclude that 
\begin{align*}
T_\infty\pair{\alpha}{\rho}\geq\calP_\infty\at{u_\zeta}=2+\bat{1+\Do{1}}\at{\frac{c\varrho^\eta}{\alpha\at{1+\eta}}\zeta^\eta-\zeta^2},
\end{align*}
and the claim follows from choosing $\zeta$ sufficiently small.
\end{proof}
We finally establish the convergence
$T_N\pair{\alpha}{\varrho}{}\to{}T_\infty\pair{\alpha}{\varrho}$ as $N\to\infty$. To this end
we denote elements of $\ell^2_N$ by $u_N=\at{u_\dblprm{N}{j}}_j$ and introduce two operators
$R_N$ and $E_N$ which act on profiles $u$ as follows. $R_N{u}$ is the restriction
of $u$ to the symmetrized periodicity cell $Z_N$, i.e.,
\begin{align}
\label{Eqn:DefR}
\at{R_Nu}_j=
\begin{cases}
u_j&\text{for $j\in\tilde{Z}_N$},\\
0&\text{otherwise},
\end{cases}
\end{align}
and $E_N{u}$ is defined as the periodic continuation of $R_N{u}$. Notice that we allow for
small embedding errors. In fact, for non-symmetric periodicity cells with
$\tilde{Z}_N\neq{Z_N}$ we have ${E}_Nu_N\neq{u_N}$ for $u_N\in\ell^2_N$, but the embedding error is small due to the decay estimate from Lemma \ref{Lem:BasicProps}.
\par
We readily verify that $E_N$ and $R_N$ are linear operators with
\begin{align*}
R_N:\calC_N\to\calC_\infty,\quad{}E_N:\calC_\infty\to\calC_N.
\end{align*}
and $\norm{R_N{u}_N}\leq{\norm{u_N}}$ and $\norm{E_N{u}_\infty}\leq{\norm{u_\infty}}$.
\begin{lemma}
\label{Lem:REProps} We have
\begin{align*}
\abs{\calP_N\at{u_{N}}-\calP_\infty\at{R_Nu_{N}}}\leq\nDO{\norm{u_N}/\sqrt{N}}
,\quad%
\calP_N\at{E_Nu_\infty}\xrightarrow{N\to\infty}\calP_\infty\at{u_\infty}
\end{align*}
for all $u_{N}\in\calC_{N}$ and $u_\infty\in\calC_\infty$.
\end{lemma}
\begin{proof}
The first claim is a consequence of \eqref{Eqn:Misc1} and Lemma \ref{Lem:BasicProps}. The
second one follows from $R_Nu_\infty\to{u_\infty}$ as $N\to\infty$.
\end{proof}
We are now able to prove the convergence of maxima.
\begin{lemma}
\label{Lem:ConvergenceSuprema}%
We have $T_N\pair{\alpha}{\varrho}\to{}T_\infty\pair{\alpha}{\varrho}$ as $N\to\infty$.
\end{lemma}
\begin{proof}
For fixed $u_\infty\in\calC_\infty\cap\calS_\dblprm{\infty}{\varrho}$, Lemma \ref{Lem:TProps1}
and Lemma \ref{Lem:REProps} provide
\begin{align*}
\calP_\infty\nat{u_\infty}
&\leq{\calP_N\nat{E_Nu_\infty}}+\nDo{1}
\leq%
{\alpha\varrho}T_N\pair{\alpha}{\norm{E_Nu_\infty}^2}+\nDo{1}
\\&%
\leq{\alpha\varrho}T_N\pair{\alpha}{\varrho}+\nDo{1},
\end{align*}
with $\Do{1}$ being arbitrary small for large $N$. Passing to the limit $N\to\infty$ and
taking the supremum over $u_\infty$ gives
\begin{math}
T_\infty\pair{\alpha}{\varrho}\leq\liminf_{N\to\infty}T_N\pair{\alpha}{\varrho}. %
\end{math}
Towards the reverse estimate, we denote the maximizer of $\calP_N$ in
$\calC_N\cap\calS_\dblprm{N}{\varrho}$ by $\hat{u}_N$, and thanks to Lemma \ref{Lem:REProps}
we find
\begin{align*}
{\alpha\varrho}T_N\pair{\alpha}{\varrho}=\calP_N\nat{\hat{u}_{N}}
\leq%
{\calP_\infty\nat{E_N\hat{u}_N}}+
\nDo{1}\leq{\alpha\varrho}T_{\infty}\pair{\alpha}{\varrho}+\nDo{1},
\end{align*}
which implies
$\limsup_{N\to\infty}T_N\pair{\alpha}{\varrho}\leq{}T_\infty\pair{\alpha}{\varrho}$.
\end{proof}
%
%
%
%
%
\section{Existence of standing waves}\label{sec:Waves}
%
\subsection{Periodic waves}
%
%
We now prove the existence of periodic waves with $N<\infty$. To this end we consider the
reversed gradient flow for $\calP_N$ under the constraint $\norm{u}^2=\varrho$, that is
\begin{align}
\label{Eqn:GradientFlow}
\tfrac{\dint}{\dint\tau}u=F_N\at{u}
,\quad
F_N\at{u}=\partial\calP_N\at{u}-\si_N\at{u}u,\quad
\si_N\at{u}=\frac{\skp{\partial\calP_N\at{u}}{u}}{\norm{u}^2},\quad
\end{align}
where $\tau\geq0$ denotes the flow time, and $\si_N$ is a (dynamical) Lagrangian multiplier.
Constrained gradient flows are well known in the context of ground states for Schr\"odinger equations and sometimes referred to as
\emph{imaginary time methods}, see \cite{BD04,JL07} and references therein.
\begin{lemma}
\label{Lem:FlowProps} %
The gradient flow \eqref{Eqn:GradientFlow} has the following properties:
\begin{enumerate}
\item
$\calS_\dblprm{N}{\varrho}$ is an invariant set.
\item
$\calP_N\at{u}$ is strictly increasing on each non-stationary trajectory.
\item
Each stationary point $u\in\ell^2_N$ of \eqref{Eqn:GradientFlow} is a standing wave with
frequency $\si=\si_N\at{u}$.
\end{enumerate}
\end{lemma}
\begin{proof}
The definition of $\si$ implies
\begin{math}
\tfrac{\dint}{\dint\tau}\norm{u}^2=\skp{F_N\at{u}}{u}=0
\end{math}, %
so $\norm{u}$ is conserved under the flow. By a direct computation we find
\begin{align*}
\tfrac{\dint}{\dint\tau}\calP_N\at{u}=\skp{\partial\calP_N\at{u}}{F_N\at{u}}
=%
\frac{1}{\norm{u}^2}\at{\norm{\partial\calP_N\at{u}}^2\norm{u}^2-
\abs{\skp{\partial\calP_N\at{u}}{u}}^2}\geq0
\end{align*}
with $\tfrac{\dint}{\dint\tau}\calP_N\at{u}=0$ if and only if $u$ and $\partial\calP_N\at{u}$
are collinear, i.e., $\partial\calP_N\at{u}=\si_N\at{u}u$.
\end{proof}
The next result provides a key ingredient for our existence result.
\begin{lemma}
\label{Lem:InvarianceC}%
$\calC_N$ is invariant under the gradient flow \eqref{Eqn:GradientFlow}.
\end{lemma}
\begin{proof}
Within this proof let $u\at{0}\in\calC_N$ be some fixed initial data and denote by
$u\at{\tau}\in\ell^2_N$ with $\tau\geq0$ the corresponding forward in time solution to
\eqref{Eqn:GradientFlow}. We first note that
\begin{align*}
\calE_N=\{u\;:\;u_j=u_{-j}\quad\forall\quad{}j\in{Z_N}\},
\end{align*}
the set of all even profiles, is a closed linear subspace of $\ell^2_N$ and moreover invariant under $\calF_N$. We thus conclude that
$u\at{0}\in\calE_N$ implies $u\at{\tau}\in\calE_N$ for all $\tau$. In order to show the
remaining assertions we proceed as usual and consider the perturbed initial value problem
\begin{align}
\label{Lem:InvarianceC.Eqn1} %
\tfrac{\dint}{\dint\tau}v=F_N\at{v}+\eps
,\quad%
v\at{0}=u\at{0}+\eps
\end{align}
where the perturbation $\eps=\at{\eps_j}_j$ is supposed to be an inner point of $\calC_N$
with respect to the induced topology of $\calE_N$. First suppose for contradiction that there
exist $\tau_0>0$  and $j_0\in{Z_N}$ such that
\begin{align*} v_{j_0}\at{\tau_0}=0
\quad\text{and}\quad%
v_{j}\at{\tau}\geq{0} \quad\text{for all}
\quad{j\in{Z_N}}\quad\text{and}\quad 0\leq\tau<\tau_0.
\end{align*}
This implies $\tfrac{\dint}{\dint\tau}v_{j_0}\at{\tau_0}\leq0$, so 
\eqref{Lem:InvarianceC.Eqn1} provides
\begin{align*}
\tfrac{\dint}{\dint\tau}v_{j_0}\at{\tau_0}
=%
F_\dblprm{N}{j_0}\at{v\at{\tau_0}}+\eps_{j_0}=\alpha
\at{v_{j_0+1}\at{\tau_0}+v_{j_0-1}\at{\tau_0}}+\eps_{j_0}\geq\eps_{j_0}>0,
\end{align*}
which is the desired contradiction. Secondly, assuming
\begin{align*}
v_{j_0-1}\at{\tau_0}=v_{j_0}\at{\tau_0}
\quad\text{and}\quad%
v_{j-1}\at{\tau}\geq{v_{j}}
\quad\text{for all}\quad1\leq{j}\leq\max{Z_N}
\quad\text{and}\quad0\leq\tau<\tau_0,
\end{align*}
we find a contradiction by
\begin{align*}
0\geq\tfrac{\dint}{\dint\tau}v_{j_0-1}\at{\tau_0}-
\tfrac{\dint}{\dint\tau}v_{j_0}\at{\tau_0}=
\alpha\at{v_{j_0-2}+v_{j_0}-v_{j_0+1}-v_{j_0-1}}+\eps_{j_0-1}-\eps_{j_0}\geq
\eps_{j_0-1}-\eps_{j_0}>0.
\end{align*}
In conclusion, we have shown $v\at{\tau}\in\calC_N$ for all $\tau\geq0$ and all perturbations
$\eps$, and $u\at{\tau}\in\calC_N$ follows by $\norm{\eps}\to0$.
\end{proof}
\begin{theorem}
\label{Theo:PeriodicWaves}
For each $\varrho>0$ there exists a standing wave $\pair{u}{\si}$ with
$u\in\calC_N\cap\calS_\dblprm{N}{\varrho}$ and
$\si\varrho\geq{}\calP_N\at{u}=\alpha\varrho{}T_N\pair{\alpha}{\varrho}>2\alpha\varrho$.
\end{theorem}
\begin{proof}
Due to $N<\infty$ the set $\calS_\dblprm{N}{\varrho}\cap\calC_N$ is compact in $\ell^2_N$,
and hence there exist a maximizer $u$ for $\calP_N$ on this set. Lemma \ref{Lem:FlowProps}
combined with Lemma \ref{Lem:InvarianceC} then implies that $u$ is stationary point of
\eqref{Eqn:GradientFlow} and solves the standing wave equations
\eqref{Eqn:StandingWaveEquation} with frequency $\si=\si_N\at{u}$. Finally, testing
\eqref{Eqn:StandingWaveEquation} with $u$ gives
\begin{align*}
\si_N\at{u}\varrho
=%
\alpha\calL_N\at{u}+\sum\limits_{j\in{Z_N}}\Psi^\prime\at{u^2_j}u^2_j
\geq
\alpha\calL_N\at{u}+\sum\limits_{j\in{Z_N}}\Psi\at{u^2_j}=\calP_N\at{u}
\end{align*}
due to Assumption \ref{Intro:Ass}, and \eqref{Eqn:TProps3} completes the proof.
\end{proof}
We conclude this section with two remarks.
\begin{enumerate}
\item
Theorem \ref{Theo:PeriodicWaves} provides the existence of two families of periodic waves
as it holds in both the on-site and the inter-site setting.
\item
Theorem \ref{Theo:PeriodicWaves} does not exclude that the maximizer is equal to the
constant profile $\sqrt{\varrho/N}$. However, Lemma \ref{Lem:TProps2a}, combined with
Lemma \ref{Lem:ConvergenceSuprema} and \eqref{Eqn:TProps4}, ensures that the profile is
non-constant for large $N$, provided that $(A1)$, $(A2)$, or $(A3)$ is satisfied.
\end{enumerate}
%
%
\subsection{Homoclinic waves}\label{sec:HomoclinicWaves}
%
In this section we prove that the period waves from Theorem \ref{Theo:PeriodicWaves} converge
to homoclinic waves provided that the strict maximum condition \eqref{Cond:TSC} is satisfied.
To this end we fix $\varrho>0$ and consider a sequence of profiles
$\at{u_N}_N\subset\calC_\infty$ such that
\begin{enumerate}
\item
$u_N$ is the image of a maximizer of $\calP_N$ on
$\calC_N\cap\calS_{\dblprm{N}{\varrho}}$ under the restriction map $R_N$ from
\eqref{Eqn:DefR},
\item
$\si_N$ is the corresponding frequency.
\end{enumerate}
According to these definitions and Lemma \ref{Lem:BasicProps} we have
\begin{align}
\label{Eqn:SLim.Prop0}
0 \leq u\dblidx{N}{j}=\sqrt{\varrho/\at{2\abs{j}+1}}
\quad\forall\;j,\qquad\qquad
u\dblidx{N}{j}=0\quad\forall\;\abs{j}\geq\at{N+1}/2.
\end{align}
Moreover, Lemma \ref{Lem:ConvergenceSuprema} and Theorem \ref{Theo:PeriodicWaves} provide
\begin{align}
\label{Eqn:SLim.Prop1}
\si_{N}u\dblidx{N}{j}
=%
\alpha\at{u\dblidx{N}{j+1}+u\dblidx{N}{j-1}}+
\Psi^\prime\at{u\dblidx{N}{j}^2}u\dblidx{N}{j}
\quad\text{for all $j,\,N$ with $\abs{j}\leq{\at{N-1}/2}$},
\end{align}
as well as
\begin{align}
\label{Eqn:SLim.Prop2}
\si_{N}\varrho\geq\calP_\infty\at{u_N}+\Do{1}
=%
\alpha\varrho{T_\infty\pair{\varrho}{\alpha}}+\nDo{1}.
\end{align}
We next show by using the strict maximum condition that the profiles $u_N$ are localized.
\begin{lemma}
%
Suppose that \eqref{Cond:TSC} is satisfied. Then,
\begin{align*}
\ul{\sigma}=\liminf\limits_{N\to\infty}\sigma_N>2\alpha
\end{align*}
and there exist two positive constants $C$ and $d$ such that
\begin{align}
\label{Lem:SLim.Decay.Eqn1}
u\dblidx{N}{j}\leq{C}\exp\at{-d\,\abs{j}}
\end{align}
holds for all $j,\,N$.
\end{lemma}
\begin{proof}
The first claim follows from \eqref{Eqn:SLim.Prop2}.
Now choose $\sigma_\star$ and $j_\star>1$ such that
\begin{align*}
2\alpha<\si_\star<\ul\sigma,\quad\sup_{0\leq{x}
\leq%
\varrho/\at{2j_\star+1}}\Psi^\prime\at{x}\leq\bar\sigma-\sigma_\star.
\end{align*}
Combining this with \eqref{Eqn:SLim.Prop0}, \eqref{Eqn:SLim.Prop1}, and
$u\dblidx{N}{j}\geq{}u\dblidx{N}{j+1}$ gives
\begin{align*}
\at{\sigma_\star-\alpha}u\dblidx{N}{j}
&\leq%
\at{\sigma_N-\Psi^\prime\at{u\dblidx{N}{j}^2}- \alpha}
u\dblidx{N}{j}
\\&\leq%
\sigma_Nu\dblidx{N}{j}-\Psi^\prime\at{u\dblidx{N}{j}^2}
u\dblidx{N}{j}-\alpha{u\dblidx{N}{j+1}}
\leq%
\alpha{u\dblidx{N}{j-1}}
\end{align*}
and hence $u\dblidx{N}{j}\leq{\kappa}^{j-j_\ast}\sqrt{\varrho}$ with
$\kappa=\frac{\alpha}{\sigma_\ast-\alpha}<1$ for all $j$ with $j_\ast<j<N/2-2$. Finally,
\eqref{Lem:SLim.Decay.Eqn1} follows with $d=-\ln{\kappa}$ and $C$ sufficiently large.
\end{proof}
\begin{corollary}
\label{Cor:CompactnessMaxSequence}
Suppose that $\alpha$ and $\varrho$ are chosen such that \eqref{Cond:TSC} is satisfied. Then, the sequence $\at{u_N}_N$ is strongly
compact and each accumulation point $u_\infty$ satisfies $u_\infty\in\calC$,
$\norm{u_\infty}^2=\varrho$, and
$\calP_\infty\at{u_\infty}=\alpha\varrho{T}_\infty\pair{\alpha}{\varrho}$. Moreover,
$u_\infty$ decays exponentially and is a standing wave with frequency $\si_\infty>2\alpha$.
\end{corollary}
\begin{proof}
By compactness we can extract a (not relabelled) subsequence such that
$u_{N}\rightharpoonup{}u_\infty\in\calC$ in $\ell^2_\infty$, and this yields the pointwise 
convergence $u\dblidx{N}{j}\to{}u\dblidx{\infty}{j}$ for all $j$. The uniform tail estimate
\eqref{Lem:SLim.Decay.Eqn1} then implies $\norm{u_\infty}^2=\varrho$ and that $u_\infty$
decays exponentially for $j\to\pm\infty$. We conclude that $u_{N}\to{}u_\infty$ strongly in
$\ell^2_\infty$ as well as
\begin{math}
\calP_\infty\at{u_\infty}=\lim_{N\to\infty}\calP_\infty\at{u_{N}}=
\alpha\varrho{T}\pair{\alpha}{\varrho}>2\alpha\varrho,
\end{math} %
where we used \eqref{Eqn:SLim.Prop2}. Moreover, \eqref{Eqn:SLim.Prop1} with fixed $j$ and
\eqref{Eqn:SLim.Prop2} imply $\si_{N}\to\si_\infty$ for some $\si_\infty>2\alpha$, and
exploiting \eqref{Eqn:SLim.Prop1} for all $j$ we infer that $u_\infty$ is a standing wave
with frequency $\si_\infty$.
\end{proof}
We have now finished the existence proof for standing waves. In particular, Theorem
\ref{Intro:Result} follows from Lemma \ref{Lem:TProps2a}, Lemma \ref{Lem:TProps2b}, Theorem \ref{Theo:PeriodicWaves}, and Corollary \ref{Cor:CompactnessMaxSequence}. We finally recall that 
for given $\alpha>0$ the condition $\at{A3}$ on $\Psi$ implies \eqref{Cond:TSC} for all $\varrho>0$,
and hence the existence of homoclinic waves with arbitrary small energy.
%
%
%
\section{Numerical examples}\label{sec:ApproxWaves}
%
%
In this section we illustrate our analytical results by numerical simulations of standing
waves with $N<\infty$. To this end we define a map
$I:\calS_\dblprm{N}{\varrho}\to\calS_\dblprm{N}{\varrho}$ by
\begin{align*}
I\at{u}=\sqrt{\varrho}\frac{{u}+\tau{}F_N\at{u}}{\norm{{u}+\tau{}F_N\at{u}}},
\end{align*}
where $\tau>0$ is sufficiently small, and construct standing waves as limits of the iteration
\begin{align}
\label{Eqn:NumScheme}
u_0=u_\ini
,\quad
u_{k+1}=I\at{u_k}.
\end{align}
This scheme preserves the constraint $\norm{u}^2=\varrho$ exactly, and is a discrete analogue
to the gradient flow \eqref{Eqn:GradientFlow} due to
$I\at{u}=u+\tau{F_N}\at{u}+\nDo{\tau^2}$. In order to compute a good guess for the initial
profile $u_\ini$ we start with the ansatz
\begin{align*}
u\dblidx{\ini}{j}
=%
\ka_1+\ka_2\chi_j+\ka_3\at{1+\cos\at{\pi{j/N}}}+\ka_4\exp\bat{-20\at{{j/N}}^2}
,\quad
\chi_j=
\begin{cases}
1&\text{for} \abs{j}<1,\\0&\text{otherwise},
\end{cases}%
\end{align*}
where the parameters $\ka_i$ are positive and coupled by the constraint
$\norm{u_\ini}^2=\varrho$. Then we sample the set of all admissible parameters by about $100$
points, and solve the discrete maximization $\calP_n\at{u_\ini\at{\ka_i}}\to\max$ to find the
optimal values for the parameters $\ka_i$.
\begin{figure}[ht!]
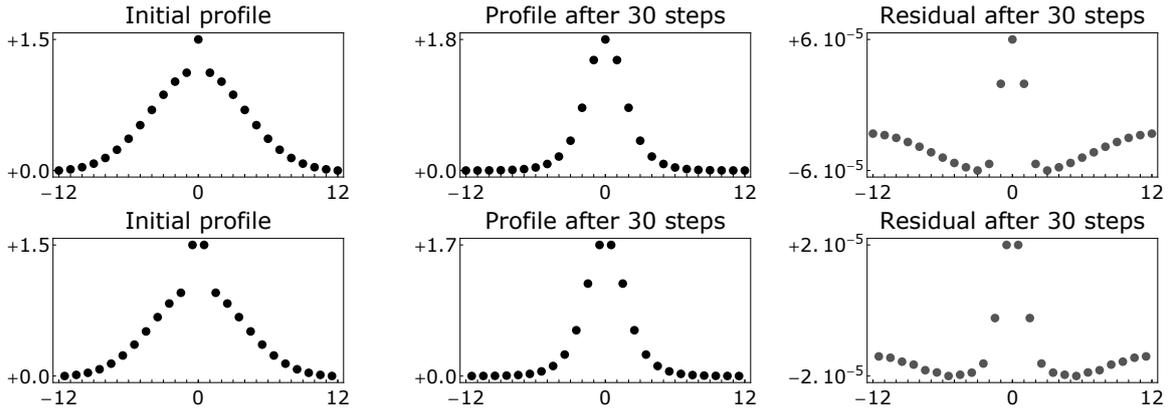
%
\centering{%
\begin{minipage}[c]{0.3\textwidth}%
\includegraphics[width=\textwidth, draft=\figdraft]%
{\figfile{ex_1a_ini}}%
\end{minipage}%
\hspace{0.025\textwidth}%
\begin{minipage}[c]{0.3\textwidth}%
\includegraphics[width=\textwidth, draft=\figdraft]%
{\figfile{ex_1a_prof}}%
\end{minipage}%
\hspace{0.025\textwidth}%
\begin{minipage}[c]{0.3\textwidth}%
\includegraphics[width=\textwidth, draft=\figdraft]%
{\figfile{ex_1a_res}}%
\end{minipage}%
\\%
\begin{minipage}[c]{0.3\textwidth}%
\includegraphics[width=\textwidth, draft=\figdraft]%
{\figfile{ex_1b_ini}}%
\end{minipage}%
\hspace{0.025\textwidth}%
\begin{minipage}[c]{0.3\textwidth}%
\includegraphics[width=\textwidth, draft=\figdraft]%
{\figfile{ex_1b_prof}}%
\end{minipage}%
\hspace{0.025\textwidth}%
\begin{minipage}[c]{0.3\textwidth}%
\includegraphics[width=\textwidth, draft=\figdraft]%
{\figfile{ex_1b_res}}%
\end{minipage}%
}%
\caption{%
Periodic on-site (top row, $N=25$) and inter-site (bottom row, $N=24$) waves for
the data from \eqref{Prm:NumEx1} with $\tau=1$.
}%
\label{Fig:NumEx1}%
\end{figure}%
\begin{figure}[ht!]
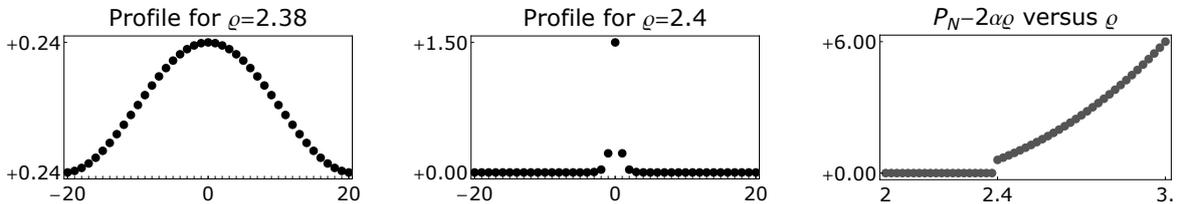
%
\centering{%
\begin{minipage}[c]{0.3\textwidth}%
\includegraphics[width=\textwidth, draft=\figdraft]%
{\figfile{ex_2_prof_1}}%
\end{minipage}%
\hspace{0.025\textwidth}%
\begin{minipage}[c]{0.3\textwidth}%
\includegraphics[width=\textwidth, draft=\figdraft]%
{\figfile{ex_2_prof_2}}%
\end{minipage}%
\hspace{0.025\textwidth}%
\begin{minipage}[c]{0.3\textwidth}%
\includegraphics[width=\textwidth, draft=\figdraft]%
{\figfile{ex_2_energy}}%
\end{minipage}%
}%
\caption{%
Periodic on-site waves for the data from \eqref{Prm:NumEx2} with
$\tau=1$ and $N=41$.
}%
\label{Fig:NumEx2}%
\end{figure}%
\bigpar%
Figure \ref{Fig:NumEx1} shows numerical results for
\begin{align}
\label{Prm:NumEx1}
\Psi\at{x}=x-\arctan{x},\quad\alpha=1,\quad\varrho=10,
\end{align}
and provides evidence that the algorithm \eqref{Eqn:NumScheme} produces a standing wave in
both the on-site and inter-site setting. Notice that $\Psi$ satisfies Assumption
\ref{Intro:Ass} and is saturable due to $\lim_{x\to\infty}\Psi^\prime\at{x}=1$.
\par
A second example concerns
\begin{align}
\label{Prm:NumEx2}
\Psi\at{x}
=%
\exp\at{x}-\tfrac{1}{2}x^2-x-1,\quad\alpha=1,\quad\varrho\in\ccinterval{2}{3}.
\end{align}
and is shown in Figure \ref{Fig:NumEx2}. The simulations indicate that the periodic on-site waves for
\eqref{Prm:NumEx2} exhibit quite different properties for small and large values of
$\varrho$: For $\varrho\leq2.38$ we observe that almost all lattice sites are excited
and that the profile
has small amplitude and is almost constant. Moreover, 
the energy $P_N=\calP_N\at{u_N}$ is only slightly larger than $2\alpha\varrho$. For
$\varrho\geq2.4$, however, the profile is strongly localized and $P_N$ is considerably
larger than $2\alpha\varrho$. We therefore expect that for $N\to\infty$ the periodic
waves converge pointwise to zero and a non-trivial homoclinic wave for small and large
$\varrho$, respectively. Notice that this is in accordance with our theoretical results: 
Since we have $\Psi\at{x}\sim\tfrac{1}{6}x^3$ for small $x$, the existence of homoclinic waves is
guaranteed only for sufficiently large $\varrho$.
\begin{figure}[ht!]
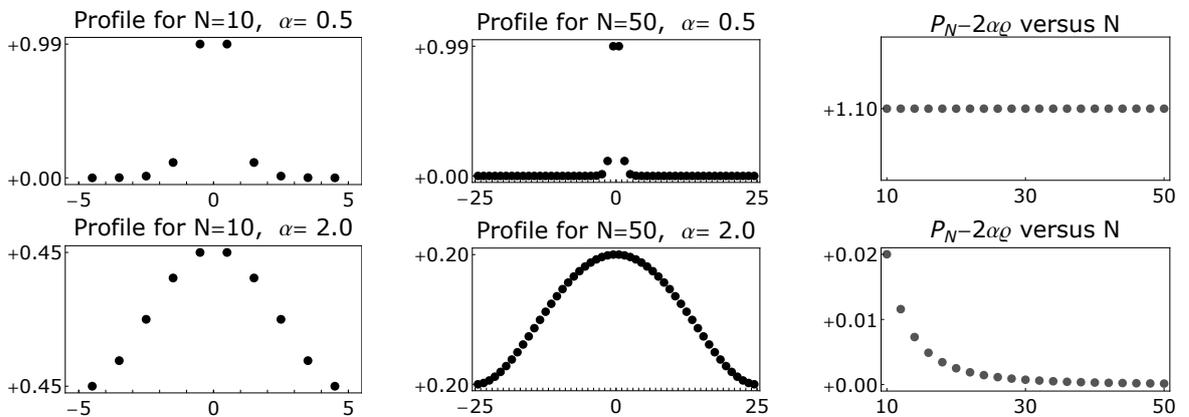
%
\centering{%
\begin{minipage}[c]{0.3\textwidth}%
\includegraphics[width=\textwidth, draft=\figdraft]%
{\figfile{ex_3a_prof1}}%
\end{minipage}%
\hspace{0.025\textwidth}%
\begin{minipage}[c]{0.3\textwidth}%
\includegraphics[width=\textwidth, draft=\figdraft]%
{\figfile{ex_3a_prof2}}%
\end{minipage}%
\hspace{0.025\textwidth}%
\begin{minipage}[c]{0.3\textwidth}%
\includegraphics[width=\textwidth, draft=\figdraft]%
{\figfile{ex_3a_energy}}%
\end{minipage}%
\\%
\begin{minipage}[c]{0.3\textwidth}%
\includegraphics[width=\textwidth, draft=\figdraft]%
{\figfile{ex_3b_prof1}}%
\end{minipage}%
\hspace{0.025\textwidth}%
\begin{minipage}[c]{0.3\textwidth}%
\includegraphics[width=\textwidth, draft=\figdraft]%
{\figfile{ex_3b_prof2}}%
\end{minipage}%
\hspace{0.025\textwidth}%
\begin{minipage}[c]{0.3\textwidth}%
\includegraphics[width=\textwidth, draft=\figdraft]%
{\figfile{ex_3b_energy}}%
\end{minipage}%
}%
\caption{%
Periodic inter-site waves
for the data from \eqref{Prm:NumEx3} with $\tau=1$ and several values for $N$.
Top and bottom row correspond to  $\alpha=0.5$ and $\alpha=2.0$, respectively.
}%
\label{Fig:NumEx3}%
\end{figure}%
\par
A similar phenomenon can be observed in Figure \ref{Fig:NumEx3}, which illustrates the limit
$N\to\infty$ for 
\begin{align}
\label{Prm:NumEx3}
\Psi\at{x}=x^4,\quad\alpha\in\{1/2,\,2\},\quad\varrho=2\,.
\end{align}
For sufficiently large $\alpha$ we have $P_N\to2\alpha\varrho$ as $N\to\infty$ and the periodic waves converge (weakly in
$\ell^2$) to zero. If $\alpha$ is sufficiently small, however, we have
$\lim_{N\to\infty}\si_N>2\alpha\varrho$ and the periodic waves converge (strongly in
$\ell^2$) to a non-trivial homoclinic wave.
%
%
\section*{Acknowledgements}%
I am very grateful to the referee for the constructive criticism which allowed me to
improve the results and the exposition.
This work was supported by the EPSRC
Science and Innovation award to the Oxford Centre for
Nonlinear PDE (EP/E035027/1).%
%
%
\providecommand{\bysame}{\leavevmode\hbox to3em{\hrulefill}\thinspace}
\providecommand{\MR}{\relax\ifhmode\unskip\space\fi MR }
\providecommand{\MRhref}[2]{%
  \href{http://www.ams.org/mathscinet-getitem?mr=#1}{#2}
}
\providecommand{\href}[2]{#2}

\end{document}